\documentclass[a4paper,UKenglish,cleveref, autoref, thm-restate, nolineno]{socg-lipics-v2021}

\hideLIPIcs

\usepackage{xspace}

\DeclareMathOperator{\Comp}{Comp}
\DeclareMathOperator{\Cross}{Cross}
\DeclareMathOperator{\Change}{Change}

\newcommand{\gcal}{\mathcal{G}}
\newcommand{\ecal}{\mathcal{E}}

\newcommand{\bigO}{\mathcal{O}}
\newcommand{\Na}{\mathbb{N}}
\newcommand{\LCR}{LCR\xspace}
\newcommand{\LCRlong}{\textsc{Layered Connectivity Reconfiguration}\xspace}
\newcommand{\LCSR}{LCSR\xspace}
\newcommand{\LCSRlong}{\textsc{Layered Connectivity Shortest Reconfiguration}\xspace}
\newcommand{\STSR}{STSR\xspace}
\newcommand{\STSRlong}{\textsc{Spanning Tree Sequence Reconfiguration}\xspace}
\newcommand{\BSR}{\textsc{BSR}\xspace}
\newcommand{\BSRlong}{\textsc{Basis Sequence Reconfiguration}\xspace}

\newcommand{\SC}{\textsc{Set Cover}\xspace}
\newcommand{\PTIME}{\textsf{P}}
\newcommand{\NP}{\textsf{NP}}

\title{Temporal Graph Reconfiguration for Always-Connected Graphs} %

\author{Paul Sievers}
    {Hasso Plattner Institute, University of Potsdam, Germany \and \url{}}
    {paul.sievers@student.hpi.de}
    {}
    {}

\author{George Skretas}
{Hasso Plattner Institute, University of Potsdam, Germany \and \url{}}
{georgios.skretas@hpi.de}
{https://orcid.org/0000-0003-2514-8004}
{}

\author{Georg Tennigkeit}
{Hasso Plattner Institute, University of Potsdam, Germany \and \url{}}
{georg.tennigkeit@hpi.de}
{https://orcid.org/0000-0003-0734-0684}
{HPI Research School on Foundations of AI (FAI)}

\authorrunning{P. Sievers, G. Skretas and G. Tennigkeit} %

\Copyright{Jane Open Access and Joan R. Public} %

\ccsdesc[500]{Mathematics of computing~Graph theory}
\ccsdesc[500]{Theory of computation~Graph algorithms analysis}

\keywords{temporal graphs, reconfiguration, layered networks, network redesign} %

\EventEditors{John Q. Open and Joan R. Access}
\EventNoEds{2}
\EventLongTitle{42nd Conference on Very Important Topics (CVIT 2016)}
\EventShortTitle{CVIT 2016}
\EventAcronym{CVIT}
\EventYear{2016}
\EventDate{December 24--27, 2016}
\EventLocation{Little Whinging, United Kingdom}
\EventLogo{}
\SeriesVolume{42}
\ArticleNo{23}

\begin{document}

\maketitle

\begin{abstract}

Network redesign problems ask for modifications to the edges of a given graph to satisfy certain properties. In temporal graphs, where edges are only active at certain times, we are sometimes only allowed to modify when the edges are going to be active. In practice, we might not even be able to perform all of the necessary modifications at once; changes must be applied step-by-step while the network is still in operation, meaning that the network must continue to satisfy some properties. To initiate a study in this area, we introduce the class of temporal graph reconfiguration problems.

As a starting point, we consider the \LCRlong (\LCR) problem: Given two always-connected temporal graphs $\gcal_1$ and $\gcal_2$, determine if it is possible to transform $\gcal_1$ into $\gcal_2$ by changing the time at which a single temporal edge is active in each step, such that every intermediate temporal graph is always-connected. We provide a dynamic programming algorithm for the \LCR problem. We also show that finding the shortest reconfiguration sequence between two temporal graphs cannot be approximated in polynomial time within a factor of $c \log M$ for some $c > 0$, unless $\PTIME=\NP$. Additionally, we show that the \LCR problem is equivalent to the \STSRlong (\STSR) problem introduced in \cite{DBLP:conf/isaac/HanakaIKOS24}. Therefore, our results also answer two questions raised by the authors: (i) find a simpler algorithm for the \STSR problem, (ii) show that the \STSR problem is inapproximable up to a factor of $c \log M$.

\end{abstract}

\newpage

\section{Introduction}

\emph{Network redesign problem} is an umbrella term for a fundamental class of problems in network science. Intuitively, given a graph $G$, change the edges of $G$ such that the new graph $G'$ satisfies some properties. Some examples of such properties are low diameter, low stretch, or strong connectivity. Several variants of this problem have been studied in programmable matter \cite{DBLP:conf/sirocco/KostitsynaLS24,DBLP:journals/jcss/MichailSS19}, reconfigurable networks \cite{DBLP:journals/algorithmica/AkitayaADDDFKPP21,DBLP:conf/icalp/GmyrHSS17,DBLP:journals/jcss/MichailSS19,DBLP:journals/dc/MichailSS22} and temporal graphs \cite{DBLP:conf/soda/BaralKMRWZ26,DBLP:conf/ijcai/DeligkasES23,DBLP:journals/iandc/DeligkasP22,DBLP:journals/jcss/EnrightMMZ21}. 

Of particular interest to us are network redesign problems in temporal graphs. A temporal graph is a graph augmented by adding labels to the edges. These labels indicate when the edge is available. The general framework of network redesign problems in temporal graphs is that we are given an input temporal graph $\gcal$, and we want to find the minimum number of label changes that turn it into a temporal graph $\gcal'$ that satisfies a given property. For example, in \cite{DBLP:journals/jcss/EnrightMMZ21}, the authors find the minimum number of labels necessary to delete such that the reachability of the vertices of the temporal graph is upper bounded. In \cite{DBLP:journals/iandc/DeligkasP22}, the objective remains the same, but we only allow delaying labels. Finally, in \cite{DBLP:conf/ijcai/DeligkasES23}, the authors maximize the reachability of a subset of the graph's vertices by advancing or delaying the temporal edges. All of these papers assume that in the application domain, we can change any number of connections of the network to arrive at the network $\gcal'$ instantly. 

However, in many applications, such as transportation or computer networks, there exist scenarios where network connections can only be changed gradually in small batches, or even one at a time. This constraint is relevant when networks must change connections while remaining operable. For example, when train tracks are undergoing maintenance, trains are rerouted to keep the network operational. Therefore, we would want to find an order of changing these connections such that the network remains operable after each change. Such constraints are often considered in distributed problems in dynamic graphs \cite{DBLP:journals/dc/GotteHSW23,DBLP:conf/sirocco/KostitsynaLS24}, as well as reconfiguration problems in graphs \cite{DBLP:conf/esa/BousquetI0MOSW20,DBLP:conf/stacs/BousquetI0MOSW22}.

Motivated by this, we introduce a class of problems called \emph{temporal graph reconfiguration problems}. In this class of problems, we are given a starting temporal graph $\gcal_1$ and a target temporal graph $\gcal_2$ with the same vertex set $V$, along with a property $P$. The objective is to find a sequence of changes to the edges of the graph such that $\gcal_1$ transforms into $\gcal_2$ while the graph maintains the property $P$ after every single change.  The goal of this paper is to initiate the study of the class of temporal graph reconfiguration problems.

\subsection{Our Contribution}

Before introducing our contribution, we first define some necessary notation and introduce a new problem in this class. We define a $\emph{temporal graph}$ with lifetime $T$ as $\gcal = (V, \ecal)$ where $V$ is the set of vertices and $\ecal \subseteq \binom{V}{2} \times [T]$ is the set of temporal edges. We will refer to a temporal edge simply as an edge when context makes it clear. A temporal graph can be viewed as a collection of $\emph{snapshot}$s, with the $t$-th snapshot $\gcal(t) = (V, \{(e,t) \in \ecal\})$, for a $t \in [T]$, representing the static graph containing only edges active at time $t$. We call a temporal graph $\gcal$ $\emph{always-connected}$ if every snapshot of $\gcal$ is connected. 

With this, we define the \LCRlong (\LCR) problem: Given two always-connected temporal graphs $\gcal_1$ and $\gcal_2$, determine if it is possible to transform $\gcal_1$ into $\gcal_2$ by changing the time at which a single temporal edge is active in each step, such that every intermediate temporal graph is always-connected. As an example, \cref{fig:example_reconfiguration} shows such a transformation between two temporal graphs.

\begin{figure}[h!]
    \centering
    \includegraphics[width=\textwidth]{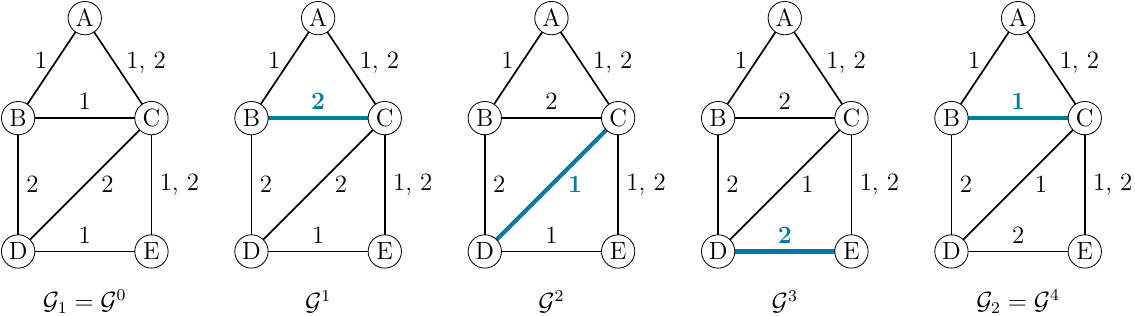}\\
    \caption{Example of a valid reconfiguration sequence $(\gcal^0, \cdots, \gcal^4)$ between $\gcal_1$ and $\gcal_2$. Time labels are shown on each edge, and edges marked in blue differ from the previous graph in the sequence.}
    \label{fig:example_reconfiguration}
\end{figure}

Our first result is a simple polynomial-time algorithm for the \LCR problem. Informally, we first define reachability partitions of bridge edges. Bridges are significant because changing their label disconnects their snapshot; however, their reachability partitions provide some crucial insights for how a bridge edge can be turned into a non-bridge edge by reconfiguring other edges. We use dynamic programming to determine, for any other edge, whether (and how) it can be turned into a non-bridge. If we identify an edge that is unchangeable no matter what reconfiguration happens before, and it has different labels in $\gcal_1$ and $\gcal_2$, reconfiguration is not possible. Otherwise, we repeatedly identify an edge that differs between $\gcal_1$ and $\gcal_2$ and apply the dynamic programming steps to change the label of that edge. The trick here is to apply every change simultaneously to $\gcal_1$ and $\gcal_2$ in order to arrive at a canonical form for both graphs. This gives a reconfiguration sequence from both $\gcal_1$ and $\gcal_2$ to the same canonical labeling; Because every reconfiguration step is reversible, reversing the sequence from $\gcal_2$ to this canonical labeling yields a reconfiguration from $\gcal_1$ to $\gcal_2$.

Our second result is an inapproximability result. We show that finding the shortest reconfiguration sequence between two temporal graphs cannot be approximated in polynomial time within a factor of $c \log M$ for some $c > 0$, unless $\PTIME=\NP$. This holds even for temporal graphs of lifetime $2$. Our reduction is from \SC. The reduction represents sets as cycles on a snapshot; choosing a set corresponds to breaking its cycle in order to close a different (more useful) cycle. To find a shortest reconfiguration sequence, one must determine the minimum number of cycles to break and reconnect.

For our final contribution, we show a relation between the \LCR problem and the \STSRlong problem introduced in \cite{DBLP:conf/isaac/HanakaIKOS24}. The \STSR problem is defined as follows: Given a multigraph $G$, a sequence of $k$ spanning trees $(T_1, \dots T_k)$ is \emph{feasible} if they are edge-disjoint. Given two such feasible sequences, the goal is to reconfigure one sequence into another: In each step, a single spanning tree $T_i$ is replaced with $T_i'=T_i-e+f$ for $e\in E(T_i)$ and $f\in E(G)\setminus E(T_i)$. We show that \LCR and \STSR are equivalent, i.e., we show (in the appendix) that a solution to the \LCR problem can be transformed into a solution to the \STSR problem, and vice versa. Because of this equivalence, our two results answer the two open questions posed by \cite{DBLP:conf/isaac/HanakaIKOS24}.

In particular, the authors provide an algorithm for the \STSR problem by first solving a more general problem, \BSRlong (\BSR), formulated in matroid theory. They pose the first open question of whether one can design a simpler algorithm for the \STSR problem. Our dynamic programming algorithm for the \LCR problem also works for the \STSR problem. It is also a simpler algorithm that relies solely on graph-theoretic arguments and the graphs' structural properties. 

Additionally, the authors in \cite{DBLP:conf/isaac/HanakaIKOS24} show that finding the shortest reconfiguration between two sequences of the \BSR problem is inapproximable up to a factor of $c\log n$. Since \BSR is a more general problem than \STSR, they pose as an open question: what is the computational complexity of the \STSR problem? We show that \STSR is also inapproximable up to the same factor as \BSR.

\subsection{Related Work}

Apart from \cite{DBLP:conf/isaac/HanakaIKOS24}, the research most closely related to ours investigates the reconfiguration of arborescences consisting of non-decreasing edge-labels in temporal graphs \cite{DBLP:conf/sand/DondiL24,DBLP:conf/wads/ItoIKKKMS23}. The key distinction between \cite{DBLP:conf/sand/DondiL24,DBLP:conf/wads/ItoIKKKMS23} and our work is that our reconfiguration changes the temporal graph itself, whereas in \cite{DBLP:conf/sand/DondiL24,DBLP:conf/wads/ItoIKKKMS23}, the authors modify a solution set on the graph.

In graph theory literature, numerous papers consider a solution set on a graph $G$ and the goal is to reconfigure it into a different solution set in the same graph $G$. In \cite{DBLP:conf/stacs/BousquetI0MOSW22}, the authors examined spanning tree reconfiguration with different constraints on degree and diameter, while in \cite{DBLP:conf/esa/BousquetI0MOSW20}, the authors considered constraints on the number of leaves. The complexity of different types of induced tree reconfiguration problems was analyzed by Wasa, Yamanaka and Arimura \cite{DBLP:journals/ieicet/WasaYA19}. In \cite{DBLP:journals/tcs/HanakaIMMNSSV20}, the authors described several reconfiguration problems in graphs under different graph structure properties, including paths, cycles, trees, and cliques, and with different operations. A polynomial-time algorithm for the matching reconfiguration problem was given in \cite{DBLP:journals/tcs/ItoDHPSUU11}. Here, the task is to transform two matchings $M$ and $M'$ into each other by adding or deleting an edge in each step, such that there is no intermediate matching of size smaller than $\min(|M|, |M'|)-1$. A generalization of this problem, allowing for more changes at each step, was analyzed by Solomon and Solomon \cite{DBLP:conf/innovations/SolomonS21}. In \cite{DBLP:conf/mfcs/BonamyBHIKMMW19}, the authors studied the perfect matching reconfiguration problem with a focus on its complexity in different graph classes and showed it to be \textsf{PSPACE}-complete.

Moving on to graph reconfiguration problems, motivated by detecting and redistricting gerrymandering, the authors in \cite{DBLP:journals/tcs/AkitayaKKST22} study the graph reconfiguration of connected graph partitions via node swapping between partitions. Deciding if there exists a reconfiguration sequence when swapping a single node at a time is \textsf{PSPACE}-complete. A plethora of papers also study graph reconfiguration problems in distributed systems, particularly in the context of programmable matter. The amoebot model considers embedded modules on a triangular grid that form a shape \cite{DBLP:journals/dc/DaymudeRS23}. The main problem investigated in the amoebot model is to develop algorithms that transform a given starting shape $A$ into another given starting shape $B$ while maintaining connectivity \cite{DBLP:conf/sirocco/KostitsynaLS24}. Other programmable matter models have also studied this question \cite{DBLP:journals/jcss/MichailSS19,DBLP:journals/jpdc/NavarraPP25}. Finally, graph reconfiguration problems have also been studied in the overlay networks literature \cite{DBLP:journals/dc/GotteHSW23,DBLP:journals/dc/MichailSS22}.

\section{Preliminaries}\label{sec:preliminaries}

For $n \in \Na$ we denote $[n] = \{1, \cdots, n\}$. For a set $X$, we write $\binom{X}{2}$ for the set of unordered pairs of elements in $X$. A temporal graph with lifetime $T$ is defined as $\gcal = (V, \ecal)$ where $V$ is the set of vertices and $\ecal \subseteq \binom{V}{2} \times [T]$ is the set of temporal edges. The $t$-th snapshot of $\gcal$ is defined as $\gcal(t) = (V, \{(e,t) \in \ecal\})$, for a $t \in [T]$, representing the static graph containing only edges active at time $t$. We define $\ecal(t) = \{e \mid (e,t) \in \ecal \}$ for $t \in [T]$. We call a temporal graph $\gcal$ always-connected if every snapshot of $\gcal$ is connected. For a temporal graph $\gcal = (V, \ecal)$ and an edge $(e,t) \in \ecal$ we denote by $\gcal - (e,t)$ the temporal graph $\gcal'=(V, \ecal \setminus \{(e,t)\})$. We write $n = |V|$ for the number of vertices and $M = |\ecal|$ for the number of temporal edges.

We say a snapshot $\gcal(t)$ is \emph{connected} if there exists a path in $\gcal(t)$ between every pair of nodes in $V$. A \emph{connected component} of $\gcal(t)$ is a subset of vertices $C \subseteq V$ such that for all pairs of nodes $a, b \in C$ there exists a path between $a$ and $b$ in $\gcal(t)$ and for all nodes $a \in C$ and $b \not \in C$ there does not exist a path between $a$ and $b$ in $\gcal(t)$. A temporal edge $(e,t) \in \ecal$ is called a \emph{bridge} if its removal increases the number of connected components of $\gcal(t)$ by one. In static graphs (and thus also in snapshots of a temporal graph), all bridges can be computed in $\bigO(|V| + |E|)$ time using the algorithm described in \cite{tarjan_note_1974}.

An \emph{edge relabeling} operation on a temporal graph takes an edge $(e, t) \in \ecal$ and changes the time it is active, replacing it with $(e, t')$ for some $t' \in [T] \setminus \{t\}$. We refer to \emph{relabeling} an edge $(e,t) \to (e, t')$ as a shorthand for this operation. A \emph{reconfiguration sequence} between temporal graphs $\gcal_1 = (V, \ecal_1)$ and $\gcal_2 = (V, \ecal_2)$ is a sequence of temporal graphs $(\gcal_1 = \gcal^0, \gcal^1, \cdots, \gcal^k = \gcal_2)$ with $\gcal^j = (V, \ecal^j)$ for $0 \le j \le k$, such that for $0 \le i < k$ the graph $\gcal^{i+1}$ can be obtained from $\gcal^i$ using a single edge relabeling operation. More formally, for all $0 \le i < k$, there exists $t, t' \in [T]$ such that for all $\tilde t \in [T] \setminus \{t, t'\}: \ecal^{i}(\tilde t) = \ecal^{i+1}(\tilde t)$ and it holds that $\ecal^i(t) \setminus \ecal^{i+1}(t) = \ecal^{i+1}(t') \setminus \ecal^{i}(t') = \{e\}$ for some $e \in \ecal^i(t)$. An \emph{intermediate graph} in such a reconfiguration sequence is any temporal graph $\gcal^i$ with $0 \le i \le k$. The length of such a reconfiguration sequence is $k$. 
We call a reconfiguration sequence valid if all intermediate graphs are always-connected. Similarly, an edge relabeling operation on an always-connected temporal graph $\gcal$ is valid if the resulting graph is always-connected.

As the edge relabeling operation does not change the number of temporal edges between two vertices, we assume that for all pairs $e \in \binom{V}{2}$ it holds that

$$|\{(e,t) \in \ecal_1 \mid t \in [T]\}| = |\{(e,t) \in \ecal_2 \mid t \in [T]\}|$$
as otherwise reconfiguration between $\gcal_1$ and $\gcal_2$ would not be possible. We also assume that any given temporal graph is always-connected.

\section{Reachability Partitions}\label{sec:reachability_partition}

In this section, we introduce important definitions and properties of a temporal graph in relation to bridges. These will serve as building blocks for our proofs in the later sections.

Since removing a bridge disconnects a snapshot, a valid edge relabeling operation can only be performed on non-bridges in order to preserve the always-connected property. To find a reconfiguration sequence between two temporal graphs $\gcal_1$ and $\gcal_2$, we need to change the edges only present in $\gcal_1$ to the ones in $\gcal_2$. Once all these different edges have been relabeled, the graphs will be equal. 
If one of the different edges is a non-bridge, it can be changed directly, which gives some progress in reconfiguration. However, as illustrated in \Cref{fig:difference_illustration}, all different edges may initially be bridges. Thus we need a systematic way to first turn these bridges into non-bridges.

\begin{figure}[h!]
    \centering
    \includegraphics[height = 4 cm]{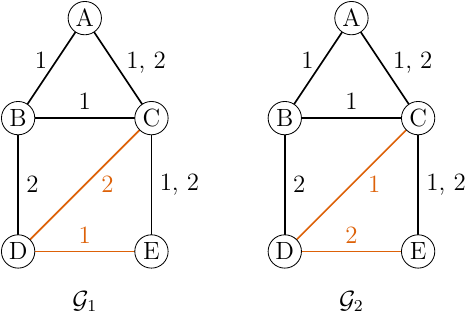}\\
    \caption{Example of temporal graphs $\gcal_1$ and $\gcal_2$ in which all different edges in $\gcal_1$ (marked in orange) are bridges.}
    \label{fig:difference_illustration}
\end{figure}

To address this, we define the $\emph{reachability partition}$ of a bridge $(\{u,v\},t)$. For a temporal graph $\gcal$, it is the partition of nodes into the two connected components of $\gcal(t) - (\{u,v\}, t)$, i.e.\ the nodes reachable from $u$ and $v$ in $\gcal(t) - (\{u,v\},t)$. See \Cref{fig:example_reachability_partition}.

\begin{restatable}[Reachability Partition]{definition}{reachability_partition}
    Given a temporal graph $\gcal = (V,\ecal)$ and a bridge $(e,t) = (\{u,v\}, t) \in \ecal$, we define the components of the reachability partition as $\Comp(u, e ,t) := \{x \in V \mid \exists \text{path between } u \text{ and } x \text{ in } \gcal(t) - (e,t)\}$ and analogously $\Comp(v,e,t) := \{x \in V \mid \exists \text{path between } v \text{ and } x \text{ in } \gcal(t) - (e,t)\}$. 
    
    We refer to the partition of $V$ into $\Comp(u, e, t)$ and $\Comp(v, e, t)$ as the reachability partition of $(e,t)$.
\end{restatable}

An edge $(\{u', v'\},t')$ is a \emph{crossing edge} in the reachability partition of a bridge $(\{u,v\}, t) = (e,t)$ if $u' \in \Comp(u,e, t)$ and $v' \in \Comp(v,e, t)$. 
Note that in an always-connected graph, $\gcal(t)$ has exactly one connected component. For a bridge $(\{u,v\}, t) = (e,t)$, it holds that $\gcal(t) - (e,t)$ has exactly two connected components. Since every node belongs to a connected component, these two connected components form a partition of $V$ and correspond to our reachability partition, as $u$ and $v$ are in different connected components. 

\begin{figure}[h!]
    \centering
    \includegraphics[height = 4 cm]{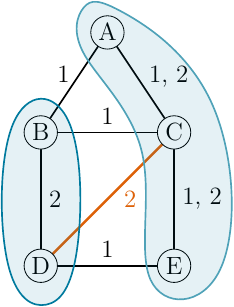}\\
    \caption{Example of the reachability partition of the edge $(e,2) = (\{C, D\},2)$.}
    \label{fig:example_reachability_partition}
\end{figure}

Using the definition of the reachability partition, we can now characterize when a bridge becomes a non-bridge after an edge relabeling operation.

\begin{restatable}[Bridges and Reachability Partition]{lemma}{cut_edges}
    \label{lm:cut_edges}
    Given a temporal graph $\gcal = (V, \ecal)$ and a bridge $(e,t)$, a valid edge relabeling operation $(e',t') \to (e', t)$ on $\gcal$ turns $(e,t)$ into a non-bridge if and only if $(e',t')$ is a crossing edge in the reachability partition of $(e,t)$.
\end{restatable}
\begin{proof}
    Let $(e, t) = (\{u,v\}, t)$ and $(e', t') = (\{x,y\}, t')$ be edges in $\gcal$. Let $\gcal'$ be the graph obtained by applying the edge relabeling operation $(e', t') \to (e', t)$ on $\gcal$. 

    Suppose $(e',t')$ is a crossing edge in the reachability partition of $(e,t)$. W.l.o.g., assume $x \in \Comp(u, e, t)$ and $y \in \Comp(v, e, t)$. 
    We show that $(e,t)$ is not a bridge in $\gcal'$ by proving that $\gcal'(t) - (e,t)$ is connected.
    
    Since $\Comp(u, e, t)$ and $\Comp(v, e, t)$ are connected components of $\gcal(t)-(e,t)$, any two nodes in the same connected component remain connected in $\gcal'(t)-(e,t)$, because $\gcal'(t) - (e,t)$ contains all edges of $\gcal(t) - (e,t)$ and the new edge $(e', t')$.
    
    For nodes $a \in \Comp(u, e, t)$ and $b \in \Comp(v, e, t)$, there exist paths $a$ to $x$ within $\Comp(u, e, t)$ and $y$ to $b$ within $\Comp(v, e, t)$ in $\gcal(t) - (e,t)$. In $\gcal'(t) - (e,t)$ the added edge $(\{x,y\},t)$ connects these paths, giving a path from $a$ to $b$. Therefore $\gcal'(t) - (e,t)$ is connected and $(e,t)$ is no longer a bridge after the edge relabeling operation.
    
    Suppose $(e', t')$ is not a crossing edge in the reachability partition of $(e,t)$. W.l.o.g., assume $x,y \in \Comp(u, e, t)$. Since $\Comp(u, e, t)$ is already connected, there exists a path $x$ to $y$ in $\gcal - (e,t)$. Thus after adding $(\{x,y\}, t)$, the connected components do not change and $\gcal' - (e,t)$ is still not connected. Therefore $(e,t)$ remains a bridge after the relabeling.
\end{proof}

Since a bridge may require multiple edge relabeling operations before it becomes a non-bridge, we need to analyze how its reachability partition changes after edge relabeling operations which leave it a bridge. This scenario is formalized in the next lemma.

\begin{restatable}[Reachability Partition invariant]{lemma}{reachability_partition_unchanged}
    \label{lm:reachability_partition_unchanged}
    Given a temporal graph $\gcal = (V,\ecal)$ and a bridge $(e,t)$, the reachability partition of $(e,t) = (\{u,v\},t )$ does not change after any valid edge relabeling operation $(e', t_1) \to (e', t_2)$ on $\gcal$ after which $(e,t)$ is still a bridge.
\end{restatable}
\begin{proof}
    Let $(e,t)$ be any bridge in $\gcal$ and $(e', t_1)$ an edge in $\gcal$. Let $\gcal'$ be the graph obtained by applying the valid edge relabeling operation $(e', t_1) \to (e', t_2)$ on $\gcal$. Since the operation is valid $(e', t_1)$ must be a non-bridge. Suppose $(e,t)$ remains a bridge in $\gcal'$.

    In $\gcal$ let $a \in \Comp(u, e, t)$. We show that $a \in \Comp(u, e, t)$ still holds in $\gcal'$. The argument for $a \in \Comp(v, e, t)$ is analogous.
    
    \textbf{Case 1:} $t_1 \neq t$ and $t_2 \neq t$. Since $\gcal(t) = \gcal'(t)$,  clearly $a \in \Comp(u, e, t)$ in $\gcal'$ holds.
    
    \textbf{Case 2:} $t_1 = t$. As the only change to $\gcal(t) - (e,t)$ is the removal of an edge, there is no new path in $\gcal'(t) - (e,t)$ between $v$ and $a$. Thus $a \not \in \Comp(v, e, t)$ still holds.  Since the components form a partition of $V$, it follows that $a \in \Comp(u, e, t)$ holds in $\gcal'$. 
    
    \textbf{Case 3:} $t_2 = t$. As the only change to $\gcal(t)-(e,t)$ is the addition of an edge, there is still a path in $\gcal'(t) - (e,t)$ between $u$ and $a$. Therefore $a \in \Comp(u, e, t)$ holds in $\gcal'$.

    As the components of the reachability partition of $(e,t)$ still form a partition of $V$ in $\gcal'$, the reachability partition does not change after the edge relabeling operation.
\end{proof}

\section{Decision Problem}\label{sec:decision_problem}

In this section, we give a polynomial-time algorithm for deciding if there exists a reconfiguration sequence between two temporal graphs such that every intermediate graph is always-connected. If such a valid reconfiguration sequence exists, the algorithm finds one of length at most $2M^2$ where $M$ is the number of temporal edges in the initial temporal graph.

The key idea is to classify edges as either \emph{changeable} or \emph{unchangeable}. A changeable edge is one that can eventually be changed, i.e.\, there exists a temporal graph, reachable through a valid reconfiguration sequence, in which it becomes a non-bridge. An unchangeable edge remains a bridge in all such reachable graphs.

With this classification, we can state our main result, which characterizes exactly when a valid reconfiguration sequence exists. We prove this result in the remainder of this section.

\begin{restatable}[Finding a reconfiguration sequence in polynomial time]{theorem}{searchProblemPolynomial}
    \label{thm:main_result}
    For temporal graphs $\gcal_1 = (V, \ecal_1)$ and $\gcal_2 = (V, \ecal_2)$, there exists a valid reconfiguration sequence between $\gcal_1$ and $\gcal_2$ if and only if all edges in $\ecal_1 \setminus \ecal_2$ are changeable. If that is the case we can find a valid reconfiguration sequence of length at most $2M^2$ in $\bigO(M^3)$ time.
\end{restatable}

To prove \Cref{thm:main_result}, we proceed in two steps. 
First, we establish \Cref{lm:change}, which shows that for any changeable edge, a shortest reconfiguration sequence to turn it into a non-bridge can be computed in polynomial time. In addition, \Cref{lm:change} allows us to identify exactly which edges are changeable.
Secondly, we prove \Cref{lm:difference}. For temporal graphs $\gcal_1$ and $\gcal_2$, the lemma establishes a way to reduce the number of different edges between them whenever there exists a changeable edge in $\gcal_1$ which is not in $\gcal_2$. 

By repeatedly applying this lemma, we obtain a valid reconfiguration sequence that transforms $\gcal_1$ and $\gcal_2$ into the same representation $\gcal_c$. This gives us a valid reconfiguration sequence between $\gcal_1$ and $\gcal_2$, proving \Cref{thm:main_result}.

\subsection*{Finding shortest reconfiguration sequences for every edge}

As it will prove useful for computing a reconfiguration sequence, we introduce a finer classification of changeability. We classify edges by the minimum number of edge relabeling operations needed to turn them into non-bridges.

\begin{restatable}[$k$-Changeable Edge]{definition}{k-changeable}
    Given a temporal graph $\gcal = (V, \ecal)$, a temporal edge $(e,t) \in \ecal$ is \emph{$k$-changeable} if there exists a valid reconfiguration sequence of length $k$ transforming $\gcal$ into a temporal graph $\gcal'$ in which $(e,t)$ is a non-bridge, and no such sequence of length $l < k$ exists. We call an edge \emph{changeable} if it is $k$-changeable for some $k \in \Na$ and \emph{unchangeable} if it is not $k$-changeable for any $k \in \Na$.
\end{restatable}

\begin{figure}[h!]
    \centering
    \includegraphics[height=4cm]{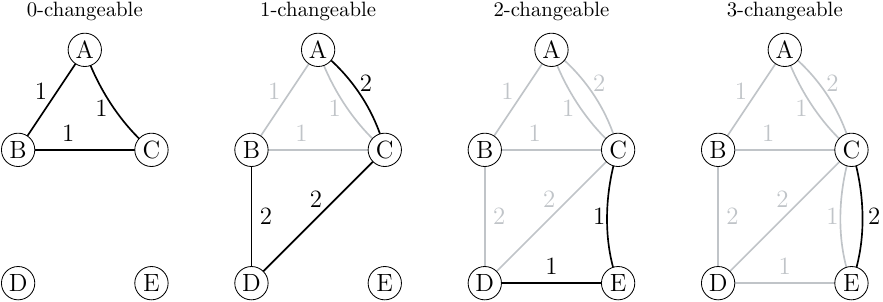}\\
    \caption{
    Illustration of $k$-changeable edges in a temporal graph. The $0$-changeable edges are exactly the non-bridges. For $i\in \Na$, the $(i+1)$-changeable edges are those bridges that become non-bridges by changing an $i$-changeable edge.
    }
    \label{fig:example_k_changeable}
\end{figure}

The following lemma uses the reachability partition to determine the number of operations needed for an edge to become changeable. It asserts that to identify all $k$-changeable edges, only $k-1$ changeable edges need to be considered, which already hints at the order of computation we will use later.

\begin{restatable}[Reachability Partition and $k$-changeability]{lemma}{reach_part_change}
    \label{lm:reach_part_change}
    Given a temporal graph $\gcal = (V, \ecal)$ and a $k \in \Na$, let $(e,t)$ be an edge that is not $l$-changeable for any $l \le k$. Then $(e,t)$ is $k+1$-changeable if and only if there exists a $k$-changeable edge $(e',t')$ that is a crossing edge in the reachability partition of $(e,t)$.
\end{restatable}
\begin{proof}
    Let $k \in \Na$ and let $(e,t)$ be a bridge in $\gcal$ which is not $l$-changeable for any $l \le k$.

    Suppose there exists a $k$-changeable edge $(e',t')$ which is a crossing edge in the reachability partition of $(e,t)$. As $(e', t')$ is $k$-changeable, there exists a valid reconfiguration sequence of length $k$ from $\gcal$ to a temporal graph $\gcal'$ in which $(e',t')$ is a non-bridge. Because $(e, t)$ is not $l$-changeable for any $l \le k$, it remains a bridge in every intermediate graph of this sequence. By \Cref{lm:reachability_partition_unchanged}, the reachability partition of $(e,t)$ remains unchanged throughout this sequence. Thus in $\gcal'$, $(e', t')$ is a non-bridge and a crossing edge in the reachability partition of $(e,t)$. By \Cref{lm:cut_edges}, the relabeling operation $(e',t') \to (e', t)$ turns $(e,t)$ into a non-bridge. Therefore $(e,t)$ is $k+1$-changeable.

    Suppose $(e,t)$ is $k+1$-changeable. Then there exists a valid reconfiguration sequence $(\gcal = \gcal^0, \gcal^1, \cdots, \gcal^{k+1})$ of length $k+1$ in which $(e,t)$ becomes a non-bridge in $\gcal^{k+1}$. Since $(e,t)$ is not $l$-changeable for any $l \le k$, it remains a bridge in all intermediate graphs $\gcal^0 \cdots \gcal^k$. By \Cref{lm:reachability_partition_unchanged}, the reachability partition of $(e,t)$ is the same in all these graphs. By \Cref{lm:cut_edges}, the last step in the sequence that makes $(e,t)$ a non-bridge involves relabeling an edge $(e', t') \to (e', t)$ in $\gcal^k$, which is a crossing edge in the reachability partition of $(e, t)$. Since $(e,t)$ is not $l$-changeable for any $l \le k$, $(e', t')$ is not $l'$-changeable for any $l' < k$. As it becomes a non-bridge in $\gcal^k$ it is thus $k$-changeable. Therefore there exists a $k$-changeable edge that is a crossing edge in the reachability partition of $(e,t)$.
\end{proof}

The following lemma provides a characterization of unchangeable edges. If there exists a $k \in \Na$ with no $k$-changeable edges, there are no $r$-changeable edges for any $r > k$ as well. In that case, all edges that are not $l$-changeable for any $l < k$ are unchangeable.

\begin{restatable}[Continuity of Changeable Edges]{lemma}{continuity}
    \label{lm:continuity}
    Given a temporal graph $\gcal = (V, \ecal)$ and a $k \in \Na^+$, if $\ecal$ can be partitioned into sets $U$ and $C$ such that every edge in $U$ is not $l$-changeable for any $l \le k$, and every edge in $C$ is $l'$-changeable for some $l' \le k - 1$, then all edges in $U$ are unchangeable.
\end{restatable}
\begin{proof}
    Let $k \in \Na^+$ such that $\ecal$ can be partitioned into sets $U$ and $C$ as in the statement of the lemma.
    We show by induction that all edges in $U$ are not $r$-changeable for any $r \in \Na$. The statement holds for all $r \le k$ by assumption.
    
    Let $r \in \Na$ with $r \ge k$ and suppose all edges in $U$ are not $l$-changeable for any $l \le r$. Let $(e,t)$ be any edge in $U$. As $r \ge k$ no edge in $C$ is $r$-changeable. Thus by \Cref{lm:reach_part_change}, $(e,t)$ is not $r+1$-changeable. This shows by induction that every edge in $U$ is not $r$-changeable for any $r \in \Na$ and therefore unchangeable.
\end{proof}

With the lemmas introduced previously, we can now develop an algorithm to compute, for any edge, the shortest reconfiguration sequence to turn it  into a non-bridge. Before describing the algorithm, we introduce a convenient structure to find crossing edges in the reachability partitions of bridges.

For a temporal graph $\gcal = (V, \ecal)$ and an edge $(e,t) \in \ecal$, we define \emph{$\Cross(e,t)$} as the set of bridges $(e', t') \in \ecal$, such that $(e, t)$ is a crossing edge in the reachability partition of $(e', t')$. The following algorithm computes $\Cross$.

For every bridge $(e', t') = (\{u',v'\}, t') \in \ecal$, first do a DFS in $\gcal(t') - (e', t')$ starting from $u'$ and from $v'$ to determine $\Comp(u',e',t')$ and $\Comp(v',e',t')$. Then, for each edge $(e,t) \in \ecal$, add $(e',t')$ to $\Cross(e,t)$ if $(e,t)$ is a crossing edge in the reachability partition of $(e',t')$.

Correctness follows directly from the definition of a crossing edge in the reachability partition. Each DFS identifies all nodes reachable from the endpoints of a bridge, giving its reachability partition. As the algorithm iterates over all edges for every bridge, $\Cross$ is computed correctly. Since there are at most $M$ bridges, and for each bridge the DFS and the iteration over all edges both take linear time, the algorithm runs in $\bigO(M^2)$ time, and therefore in polynomial time.

Now, to compute a shortest reconfiguration sequence to change any temporal edge, we define \emph{$\Change(k)$} for $0 \le k \le m$. The set $\Change(0)$ contains all non-bridges. For $0 < k \le M$, the set $\Change(k)$ contains all $k$-changeable edges, each with an index pointing to an edge in $\Change(k-1)$ that must be changed before it. The following algorithm computes $\Change$ (refer to \Cref{fig:example_k_changeable} for an example execution).

Using Tarjan's Algorithm to find all bridges, first compute $\Change(0)$. Then, starting with $k=0$, iteratively increase $k$ until $\Change(k)$ is empty. In each iteration, go through all edges $(e,t)$ in $\Change(k)$, and for each, iterate through all edges $(e', t')$ in $\Cross(e,t)$. If $(e', t')$ is not in $\Change(l)$ for any $l \le k$, we append it to $\Change(k+1)$, together with a back reference to $(e,t)$. Once $\Change(k)$ is empty, any edge not in $\Change(l)$ for any $l \le k$ is unchangeable.

\begin{restatable}[$\Change$ is computable in Polynomial Time]{lemma}{change_computable}
    \label{lm:change}
    For a temporal graph $\gcal = (V, \ecal)$ the above algorithm correctly computes $\Change(k)$ for all $0 \le k \le M$ in $\bigO(M^2)$ time.
\end{restatable}
\begin{proof}
    First, we show by induction that $\Change(k)$ contains all $k$-changeable edges for any $k \in \Na$. As Tarjan's Algorithm correctly identifies all non-bridges, this holds for $\Change(0)$.
    
    Suppose for some $i \in \Na$ that $\Change(j)$ contains all $j$-changeable edges for all $j \le i$. As we iterate through all edges $(e', t')$ in $\Change(i)$ and append all edges $(e,t)$ in $\Cross(e', t')$ which are not in $\Change(j)$ for any $j \le i$ to $\Change(i+1)$, it follows from \Cref{lm:reach_part_change} that $\Change(i+1)$ will contain exactly the $(i+1)$-changeable edges.

    When the algorithm terminates at some $i$, \Cref{lm:continuity} implies that all edges for which no sequence to change them was found, are unchangeable. Since $\Change$ is correct, we can reconstruct a shortest reconfiguration sequence to make any edge into a non-bridge using the back references stored in it.

    $\Cross$ can be computed in $\bigO(M^2)$ time, and for every edge $(e,t)$, $\Cross(e,t)$ contains at most $M$ entries. Computing $\Change(0)$ takes $\bigO(M)$ time. Since each edge $(e,t)$ is added to some $\Change(k)$ at most once, we iterate through each entry of $\Cross(e,t)$ at most once. Therefore the total running time is in $\bigO(M^2)$.
 \end{proof}

In \cref{lm:change} we showed that, for each edge, a shortest reconfiguration sequence to turn it into a non-bridge can be computed in polynomial time, if one exists. This allows us to identify which edges are changeable and which are unchangeable, completing the first part for proving \cref{thm:main_result}.

\subsection*{Finding a reconfiguration sequence to decrease difference}

Suppose we want to find a valid reconfiguration sequence from temporal graph $\gcal_1 = (V, \ecal_1)$ to $\gcal_2 = (V, \ecal_2)$. We define the \emph{difference} between them, denoted by $\delta(\gcal_1, \gcal_2)$, as the number of edges in $\gcal_1$ that are not in $\gcal_2$. Recall our assumption that $\gcal_1$ and $\gcal_2$ have the same number of edges between each pair of vertices. Under this assumption, the graphs are equal if their difference is $0$. Given the existence of a changeable edge in $\gcal_1$ that is not in $\gcal_2$, the following algorithm computes reconfiguration sequences $(\gcal_1 = \gcal_1^0, \gcal_1^1, \cdots \gcal_1^{k+1})$ and $(\gcal_2 = \gcal_2^0, \gcal_2^1, \cdots \gcal_2^k)$ for some $k \in \Na$ such that $\delta(\gcal_1^{k+1}, \gcal_2^k) < \delta(\gcal_1, \gcal_2)$.

First, compute $\Change(k)$ for all $0 \le k \le M$. Let $(e,t)$ be an edge in $\ecal_1 \setminus \ecal_2$ with a shortest valid reconfiguration sequence $(\gcal_1 = \gcal_1^0, \gcal_1^1, \cdots \gcal_1^k)$ turning it into a non-bridge. If $(e,t)$ is already a non-bridge, then $k = 0$ and the sequence consists only of $\gcal_1$. Otherwise, apply the same edge relabeling operations used in the reconfiguration sequence $(\gcal_1, \cdots, \gcal^k_1)$ to the temporal graph $\gcal_2$, obtaining $(\gcal_2 = \gcal_2^0, \gcal_2^1, \cdots \gcal_2^k)$. Finally, from $\gcal_1^k$ to $\gcal_1^{k+1}$, relabel $(e,t)$ to an edge that exists in $\gcal_2^k$ but not in $\gcal_1^k$.

It may seem unintuitive to apply relabeling operations intended for $\gcal_1$ directly to $\gcal_2$. However, we will see that our choice of $(e, t)$ guarantees that these operations are also valid when applied to $\gcal_2$.

\begin{restatable}[Sequence for decreasing difference]{lemma}{difference}
    \label{lm:difference}
    For temporal graphs $\gcal_1$ and $\gcal_2$, the above algorithm correctly computes valid reconfiguration sequences $(\gcal_1, \cdots, \gcal_1^{k+1})$ and $(\gcal_2, \cdots, \gcal_2^{k})$ for some $k \in \Na$ with $\delta(\gcal_1^{k+1}, \gcal_2^k) < \delta(\gcal_1, \gcal_2)$ in $\bigO(M^2)$ time.
\end{restatable}
\begin{proof}
    Since $\Change$ is computed correctly, the reconfiguration sequence $(\gcal_1, \cdots, \gcal_1^{k})$ is valid, and $(e,t)$ is a non-bridge in $\gcal_1^k$. Under our assumption that $\gcal_1$ and $\gcal_2$ have the same number of edges between any pair of vertices, and since $(e,t)$ is not in $\gcal_2$, there must exist an edge $(e, t')$ in $\gcal_2^k$ which is not in $\gcal_1^k$. Therefore, the whole reconfiguration sequence $(\gcal_1, \cdots, \gcal_1^{k+1})$ is valid.
    
    We show by induction that the reconfiguration sequence $(\gcal_2, \cdots, \gcal_2^k)$ is valid. For the base case, the sequence consisting only of $\gcal_2$ is trivially valid. Suppose $(\gcal^0_2, \dots, \gcal_2^j)$ is valid for some $j < k$. Let $(e', t')$ be the edge changed from $\gcal_1^j$ to $\gcal_1^{j+1}$. Since $(e,t)$ was chosen with the shortest reconfiguration sequence and $j <k$, all edges in $\gcal^j_1$ that are not in $\gcal_2^j$ are bridges. By contraposition, all non-bridges in $\gcal_1^j$ are also in $\gcal_2^j$. Every non-bridge in $\gcal_1^j$ must be part of a cycle consisting of non-bridges. Therefore a cycle including $(e', t')$ in $\gcal_1^j$ must also be in $\gcal_2^{j}$, which shows that $(e', t')$ is a non-bridge in $\gcal^j_2$. Therefore the reconfiguration sequence $(\gcal^0_2, \cdots, \gcal_2^{j+1})$ is valid. By induction, the whole reconfiguration sequence $(\gcal_2, \cdots \gcal_2^k)$ is valid.

    Since both sequences apply the same operations up to step $k$, it holds that $\delta(\gcal_1, \gcal_2) = \delta(\gcal_1^k, \gcal_2^k)$. As $(e,t)$ is not in $\gcal_2^k$ and is relabeled in $\gcal^k_1$ to an edge contained in $\gcal_2^k$, it follows that $\delta(\gcal_1^{k+1}, \gcal_2^k) < \delta(\gcal_1, \gcal_2)$.

    Finally, since computing $\Change$ takes $\bigO(M^2)$ time by \Cref{lm:change} and the sequences have length at most $M$, the total running time is in $\bigO(M^2)$.
\end{proof}

We can now prove our main result \cref{thm:main_result}, restated here for convenience.

\searchProblemPolynomial*

\begin{proof}
    If there exists an unchangeable edge in $\gcal_1$ that is not in $\gcal_2$, this edge remains a bridge in all graphs reachable from $\gcal_1$ through a valid reconfiguration sequence. In that case, reconfiguration is not possible.
    
    Otherwise, all edges in $\gcal_1$ that are not in $\gcal_2$ are changeable. 
    Using \Cref{lm:difference} we obtain $\gcal_1'$ and $\gcal_2'$ with a smaller difference than $\gcal_1$ and $\gcal_2$. By repeatedly applying \Cref{lm:difference}, the difference will eventually reach $0$. At that point both graphs have been reconfigured to a canonical $\gcal_c$. This gives us a valid reconfiguration sequences from $\gcal_1$ to $\gcal_c$, and from $\gcal_2$ to $\gcal_c$. Since every edge relabeling operation can be reversed, while preserving the always-connected property, we reverse the sequence $(\gcal_2, \cdots, \gcal_c)$ to obtain a valid reconfiguration sequence $(\gcal_c, \cdots, \gcal_2)$. Together with the sequence $(\gcal_1, \cdots, \gcal_c)$, we get a valid reconfiguration sequence from $\gcal_1$ to $\gcal_2$.

    Since there are at most $M$ edges in $\gcal_1$ that are not in $\gcal_2$, there are at most $M$ applications of \Cref{lm:difference}. As each application results in reconfiguration sequences of length at most $M$, both sequences $(\gcal_1, \cdots, \gcal_c)$ and $(\gcal_c, \cdots, \gcal_2)$ have at most length $M^2$. Therefore the total reconfiguration sequence has length at most $2M^2$. 
    Since each application takes $\bigO(M^2)$ time, the total running time is in $\bigO(M^3)$ and therefore in polynomial time. 
\end{proof}

\section{Inapproximability of Shortest Reconfiguration}\label{sec:hardness}

The natural next question is whether we can find a shortest reconfiguration sequence between two given graphs in polynomial time. We show \NP-hardness of the following decision problem: Given two always-connected temporal graphs $\gcal_1, \gcal_2$ and $\ell\in \Na$, determine if there is a valid reconfiguration sequence from $\gcal_1$ to $\gcal_2$ of length at most $\ell$. We call this problem \LCSRlong (\LCSR). We then strengthen this result by showing that the corresponding minimization problem is logarithmically hard to approximate.

\begin{restatable}{theorem}{LCSR_NP-hard}
    \label{thm:NP-hardness}
    \LCSRlong is \NP-hard. Furthermore, assuming $\PTIME\neq \NP$, \LCSR cannot be approximated in polynomial time within a factor of $c \log M$ for some $c > 0$. 
\end{restatable}

\begin{proof}

We derive both results of \Cref{thm:NP-hardness} via a reduction from \SC. A \SC instance consists of a finite set $U$, a family of subsets $\mathcal S \subseteq \mathcal P (U)$, and an integer $k$. A set cover of $U$ is a set $\mathcal S' \subseteq \mathcal S$ with $U = \bigcup_{S \in \mathcal S'} S$, and the problem asks whether there is a set cover of $U$ of size at most $k$. Raz and Safra \cite{RazSafra97} showed that assuming  $\PTIME\neq \NP$, there exists a constant $c^* > 0$ such that no polynomial-time algorithm approximates \SC within a factor of $c^* \log (|U|+|\mathcal S|)$.

\paragraph*{Construction.}

\begin{figure}[h!]
    \centering
    \includegraphics[width=\linewidth]{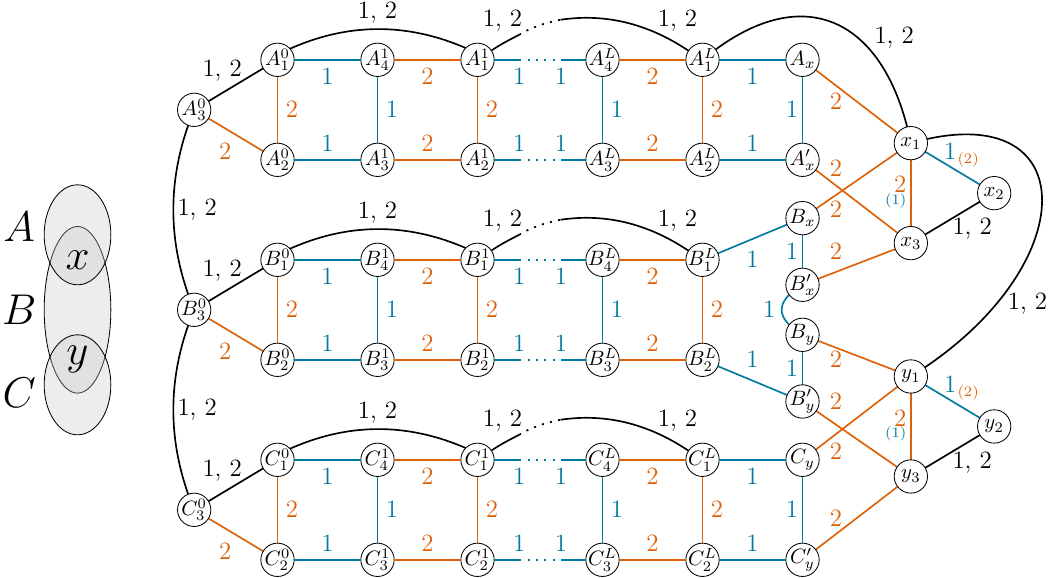}
    \caption{An example construction for a given set cover instance with elements $x,y$ and sets $A = \{x\}, B = \{x,y\}, C = \{y\}$. All labels are given for $\gcal_1$. Labels in parentheses indicate differences in $\gcal_2$, and otherwise the labels are identical.}
    \label{fig:hardness_construction}
\end{figure}

Given an instance $(U, \mathcal S, k)$ of \SC, we construct a temporal graph $\gcal_1$ with lifetime $T=2$. Whenever we define an edge with label $*$, we imply that the edge exists once in each snapshot. For clarity we initially use an integer variable $L \ge 1$ in the reduction, which we will later set to $L = |U|$.
\begin{enumerate}
    \item 
        For each $S \in \mathcal S$
        \begin{enumerate}
            \item 
                create vertices $S_1^0, S_2^0, S_3^0$ and edges $(\{S_1^0, S_2^0\}, 2), (\{S_2^0, S_3^0\}, 2), (\{S_3^0, S_1^0\}, * )$.
            \item 
                 for each $1 \le i \le L$, create vertices $S_1^i, S_2^i, S_3^i, S_4^i$ and edges $(\{S_1^i, S_2^i\}, 2), (\{S_2^i, S_3^i\}, 2), $ $(\{S_3^i, S_4^i\}, 1), (\{S_4^i, S_1^i\}, 2)$. 
            \item 
                for each $0 \le i < L$, create edges $(\{S_1^i, S_4^{i+1}\}, 1), (\{S_2^i, S_3^{i+1}\}, 1)$.
        \end{enumerate}
        We call each structure created in the first step a \emph{set-gadget}.
    \item 
        For each $u \in U$ create an \emph{element-gadget} with vertices $u_1, u_2, u_3$ and edges $(\{u_1, u_2\}, 1),$ $ (\{u_1, u_3\}, 2), (\{u_2, u_3\}, *)$.
    \item 
        For each $S \in \mathcal S$, for each $u \in S$, create vertices $S_u, S_u'$ and edges $(\{S_u, u_1\}, 2), (\{S_u', u_3\}, 2)$. %
    \item 
        For each $S \in \mathcal S$ construct a path between $S_1^L$ and $S_2^L$ through all vertices $S_u, S_u'$ for $u \in S$ such that any $S_u'$ directly follows $S_u$. Assign label $1$ to its edges.
    \item 
        Construct for the following groups of vertices a path with label $*$ that connects the group
        \begin{enumerate}
            \item
                $S_3^0$ for each $S \in \mathcal S$
            \item 
                $u_1$ for each $u \in U$
            \item 
                $S_1^i$ for each $S \in \mathcal S$, for each $0 \le i \le L$
        \end{enumerate}
        For exactly one pair of elements $S \in \mathcal S$ and $u \in U$, add the edge $(\{S_1^L, u_1\}, *)$.
        
\end{enumerate}

\noindent Refer to \Cref{fig:hardness_construction} for an example. $\gcal_2$ is identical to $\gcal_1$, except that the temporal edges in each element-gadget are flipped -- for brevity, we will refer to changing the label of an edge to the respective other label as \emph{flipping} it.

Given that $L \in \bigO(|U|)$ the graphs $\gcal_1$ and $\gcal_2$ contain $\bigO(|U|\cdot |\mathcal S|)$ vertices and edges, and can be constructed in polynomial time. Let $\ell=4  |U|+k(4L+2)$. %

We first show two lemmas that relate the sizes of solutions of the \SC instance to those of the constructed instance. From that we conclude the \NP-hardness and the logarithmic inapproximability to prove \cref{thm:NP-hardness}.

\medskip

\begin{lemma}
    \label{lm:SC_to_LCSR}
    If there exists a set cover $\mathcal S' \subseteq \mathcal S$ of $U$ of size at most $k$, then there exists a valid reconfiguration sequence between $\mathcal G_1$ and $\mathcal G_2$ of length at most $4 |U| + k(4L+2)$.
\end{lemma}
\begin{proof}
    Let $\mathcal S' \subseteq \mathcal S$ be a set cover of $U$ of size at most $k$. We reconfigure $\mathcal G_1$ into $\mathcal G_2$ as follows (see \cref{fig:hardness_reconf_example}):
    \begin{figure}[h]
        \centering
        \includegraphics[width=\linewidth]{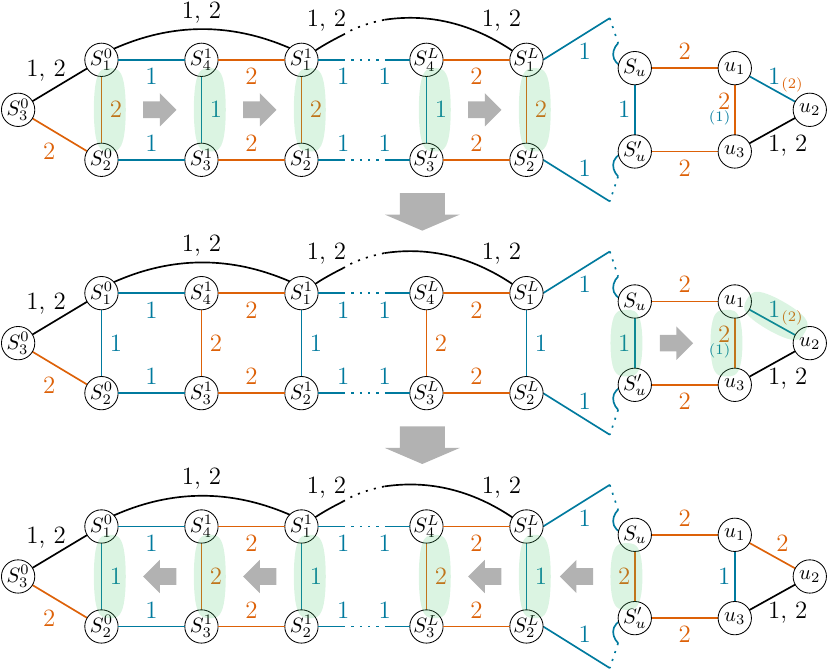}
        \caption{
        For each $S \in \mathcal  S'$, first flip the indicated edges from $(\{S_1^0, S_2^0\}, 2)$ to $(\{S_1^L, S_2^L\}, 2)$ in order to 'activate' the set-gadget. Then, 
        flip the incorrect edges in each edge-gadget that corresponds to an element in $S$. Finally, undo flips of edges that have identical labels in $\gcal_1$ and $\gcal_2$.}
        \label{fig:hardness_reconf_example}
    \end{figure}
    \begin{itemize}
    \item For each $S \in \mathcal S$:
    \begin{itemize}
        \item 
            Flip the edge $(\{S_1^0, S_2^0\}, 2)$. This closes a $1$-labeled cycle that includes the edge $(\{A_4^1, A_3^1\}, 1)$.
        \item
            For $j = 1$ to $j = L$, first flip the edge $(\{A_3^j, A_4^j\}, 1)$, which closes a $2$-labeled cycle with the edge $(\{A_1^j, A_2^j\}, 2)$. Then flip $(\{A_1^j, A_2^j\}, 2)$, which  closes a $1$-labeled cycle with the edge $(\{A_3^{j+1}, A_4^{j+1}\}, 1)$ as long as $j +1 \le L$.
        \item 
            For each element $u \in S$ for which the element-gadget does not yet have the labels assigned in $\gcal_2$:
        \begin{itemize}
            \item 
                Flip $(\{S_u, S_u'\}, 1)$. This closes a $2$-labeled cycle containing $(\{u_1, u_3\}, 2)$.
            \item 
                Flip $(\{u_1, u_3\}, 2)$, and then flip $(\{u_1, u_2\}, 1)$. This element-gadget now has the same labels as in $\gcal_2$.
            \item 
                Flip $(\{S_u, S_u'\}, 2)$ back.
        \end{itemize}
        \item 
            For $j = L$ to $j = 1$, first flip $(\{A_1^j, A_2^j\}, 1)$ back, and then flip $(\{A_3^j, A_4^j\})$ back.
        \item 
            Flip the edge $(\{A_1^0, A_2^0\}, 1)$ back.
    \end{itemize}
\end{itemize}
    Because $S$ is a set cover, all element-gadgets will eventually be correctly reconfigured. Because these hold the only differences between $\gcal_1$ and $\gcal_2$, all other changes are reverted, and we only ever flip non-bridges, this is a valid reconfiguration sequence from $\gcal_1$ to $\gcal_2$. The sequence contains $4$ steps for every element in $U$ and $4L+2$ steps for each element of $S$. Thus it has a length of at most $4|U| + k(4L+2)$ operations.
\end{proof}

\begin{lemma}
    \label{lm:LCSR_to_SC}
    If there exists a valid reconfiguration sequence between $\mathcal G_1$ and $\mathcal G_2$ of length $\ell$, then there exists a set cover of $U$ of size at most $\lfloor \frac{\ell - 4|U|}{4L+2} \rfloor$.
\end{lemma}
\begin{proof}
    Let $\sigma$ be a valid reconfiguration sequence from $\gcal_1$ to $\gcal_2$ of length $\ell$. Let $\mathcal S' \subseteq \mathcal S$ contain all sets $S \in \mathcal S$ for which $(\{S_1^L, S_2^L\}, 2)$ or $(\{S_2^L, S_3^L\}, 2)$ were flipped in $\sigma$. We show that $\mathcal S'$ is a set cover of $U$ of size at most $\lfloor \frac{\ell -4|U|}{4L + 2} \rfloor$.

    To show that $\mathcal S'$ is a set cover of $U$, we analyze the reachability partitions of relevant edges. By \cref{lm:cut_edges} a bridge can only become a non-bridge once a crossing edge in its reachability partition is changed. Starting with an edge that is different in $\gcal_1$ and $\gcal_2$ this enables us to keep track of all edges that need to be changed before it. By \cref{lm:reachability_partition_unchanged} the reachability partitions will not change until a bridge becomes a non-bridge, so it suffices to look at the initial reachability partitions of $\gcal_1$.
    Let $u \in U$ be any element. We let $C_0(u) := \{(\{u_1, u_2\}, 1)\}$, and define $C_1(u), C_2(u)$ and $C_3(u)$ in the following, so that to change any edge in $C_i$ for $0 \le i \le 2$, an edge in $C_{i+1}$ needs to be changed first. For brevity we just write \emph{component} for the component of the reachability partition, and \emph{crossing edges} for crossing edges in the reachability partition.
    See also \Cref{fig:hardness_crossing_edges}.
    \begin{figure}[h]
        \centering
        \includegraphics[width=0.7\linewidth]{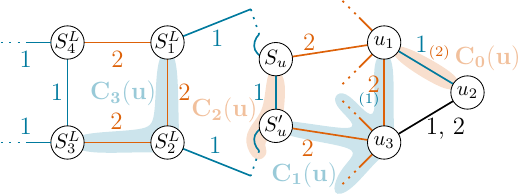}
        \caption{
         Chain of dependencies for flipping the edge $\{u_1, u_2\}$: Before an edge in $C_i$ may be flipped, at least one edge of $C_{i+1}$ must be flipped. The indicated edges related to set $S$ are representative for every set that contains $u$.}
        \label{fig:hardness_crossing_edges}
    \end{figure}
    \begin{enumerate}
        \item 
            For the edge $(\{u_1, u_2\}, 1)$ one component is $\{u_2, u_3\}$, and thus the crossing edges are $C_1(u) := \{(\{u_1, u_3\}, 2)\} \cup \{(\{u_3, S_u'\}, 2)  \mid S \in \mathcal S, u \in S\}$.
        \item 
            For the edge $(\{u_1, u_3\}, 2)$ of $C_1(u)$ one component is $\{u_2, u_3\} \cup \{S_u' \mid S \in \mathcal S,  u  \in S\}$. Therefore the crossing edges are the two incident edges of label $1$ to every vertex $S_u'$ for $S \in \mathcal S$ with $u \in S$, and the edge $(\{u_1, u_2\}, 1)$. For all other edges $(\{u_3, S_u'\}, 2)$ of $C_1(u)$ one component consists only of $S_u'$, which means only the two incident edges of label $1$ to $S_u'$ are crossing edges. We can ignore $(\{u_1, u_2\}, 1)$ as that edge already belongs to $C_0(u)$. We define $C_2(u)$ as all other crossing edges. These edges have the form $(\{S_u, S_u'\}, 1), (\{S_u', S_2^L\}, 1)$ or $(\{S_u', S_v\}, 1)$ for $S \in \mathcal S$ with $u \in S$, and $v \in S \setminus \{u\}$.
        \item 
            For any 1-labeled edge incident to vertex $S_u'$ of $C_2(u)$ one component consists of $S_2^L$ and a subset of the vertices of the $1$-labeled path between $S_1^L$ and $S_2^L$ that does not include $S_1^L$. Therefore the crossing edges are $(\{S_1^L, S_2^L\},2), (\{S_2^L, S_3^L\}, 2)$, as well as potentially $(\{S_v,  v_1\}, 2), (\{S_v',  v_3\}, 2)$  for  some $v \in S$. As the latter two edges only have edges of $C_2(u)$ as their crossing edges, we can ignore them. Thus we define $C_3(u)$ to contain all edges $(\{S_1^L, S_2^L\},2), (\{S_2^L, S_3^L\}, 2)$ for $S \in \mathcal S$ with $u \in S$.
    \end{enumerate}
    Now we show that $\mathcal S'$ is a set cover. For any element $u \in U$, the edges of $C_3(u)$ only include $(\{S_1^L, S_2^L\}, 2)$ or $(\{S_2^L, S_3^L\}, 2)$ for any $S \in \mathcal S$ with $u \in S$. One of these edges needs to have been changed in $\sigma$ for the reconfiguration sequence to be valid. Therefore it holds that $S \in \mathcal S'$ for some $S \in \mathcal S$ with $u \in S$, which shows that $\mathcal S'$ is a set cover of $U$.
        
    It remains to give an upper bound on the size of $|\mathcal S'|$, for which we give a lower bound on the number of operations $\ell$ in $\sigma$. We 
    first show that at least $4|U|$ operations are done outside set-gadgets, and afterwards show that at least $|\mathcal S'| \cdot (4L+2)$ operations are done in set-gadgets.
    
    There are $2|U|$ operations needed to change all edges that differ between $\gcal_1$ and $\gcal_2$. 
    As any other changed edge needs to be flipped back, it suffices to show that at least one additional edge is changed for every element $u \in U$ to get our lower bound of $4|U|$ operations outside set-gadgets. For any element $u \in U$, the edges that differ between $\gcal_1$ and $\gcal_2$ belong to $C_0(u)$ and $C_1(u)$. As one edge of $C_2(u)$ needs to be changed before, and the edges of $C_2(u)$ do not belong to set-gadgets, it suffices to show that a single edge of $C_2(u)$ cannot be used to fix multiple element-gadgets. To show that we argue that for vertices $u, v \in U$ with $u \neq v$ the sets $C_2(u)$ and $C_2(v)$ are disjoint. For any $S \in \mathcal S$ that includes at most one of $u$ and $v$, it is clear that no edges of the $1$-label path between $S_1^L$ and $S_2^L$ belong to both $C_2(u)$ and $C_2(v)$.
    For any $S \in \mathcal S$ such that $u \in S$ and $v \in S$, the vertices $S_u'$ and $S_v'$ have at least distance 2 on $1$-label path between $S_1^L$ and $S_2^L$. Therefore the edges of $C_2(u)$ and $C_2(v)$, which are only incident to $S_u'$ and $S_v'$ respectively cannot overlap. This shows that $C_2(u)$ and $C_2(v)$ are disjoint. Therefore an additional edge outside set-gadgets needs to be changed for every element $u \in U$, so at least $4|U|$ operations outside of set-gadgets are done in $\sigma$.

    Next we show that at least $|\mathcal S'| \cdot (4L+2)$ operations are needed in the set-gadgets. As all edges in set-gadgets are identical in $\gcal_1$ and $\gcal_2$, it suffices to show that for every $S \in \mathcal S'$, at least $2L+1$ different edges in the set-gadget of $S$ are changed, as each edge needs to be flipped back. Let $S \in \mathcal S'$. 
    We start with the edges $(\{S_1^L, S_2^L\}, 2), (\{S_2^L,S_3^L\})$ that we identified as $C_3(u)$ for some $u \in S$ before. Let $1 \le i \le L$. 
    For the edges $(\{S_1^i,  S_2^i\}, 2)$ and $(\{S_2^i,  S_3^i\}, 2)$ one component consists of $S_3^i$, possibly together with $S_2^i$. Therefore, the crossing edges are $(\{S_2^{i-1},  S_3^i\}, 1), (\{S_3^i,  S_4^i\}, 1)$, as well as additional edges that need to be changed afterwards: For $i < L$ this additional edge is $(\{S_2^i,  S_3^{i+1}\}, 1)$. For $i = L$ it is $(\{S_2^L,  S_u'\}, 1)$ for some $u \in S$. One component of this edge consists of the single vertex $S_2^L$, and thus has the edges of $C_3(u)$ as crossing edges. So it cannot be changed before as well. This shows that an element-gadget cannot be used to change a different set-gadget than the one through which it was initially reached.
    Conversely, for the edges $(\{S_3^i,  S_4^i\}, 1)$ and $(\{S_2^{i-1},  S_3^{i}\}, 1)$ the relevant crossing edges are $(\{S_3^{i-1},  S_2^{i-1}\}, 2)$ and $(\{S_2^{i-1},  S_1^{i-1}\}, 1)$ by the same argument as above. This chain of dependencies starting from $(\{S_1^L, S_2^L\}, 2), (\{S_2^L,S_3^L\})$ to the edges that are initially non-bridges $(\{S_1^0, S_2^0\}, 2), (\{S_2^0, S_3^0\}, 2)$ consists of exactly $2L+1$ edges. As shown above, this chain cannot be used to make edges of a different set-gadget into non-bridges. Therefore at least $|\mathcal S'| \cdot (4L+2)$ operations are done in the set-gadgets in $\sigma$. In total this shows for the length $\ell$ of $\sigma$ that $\ell \ge  4|U| + |\mathcal S'| \cdot (4L+2)$, which implies that $|\mathcal S'| \le \lfloor \frac{\ell - 4|U|}{4L+2} \rfloor$.
\end{proof}

To show the \NP-hardness of \LCSR we gave a polynomial-time construction from instances $(U,  \mathcal S, k)$ of set cover to instances $(\gcal_1, \gcal_2, 4|U|+k(4L+2))$ of \LCSR. If there is a set cover of size at most $k$ then by  \cref{lm:SC_to_LCSR} there is a reconfiguration sequence of size at most $4|U| + k(4L+2)$. If there is a reconfiguration sequence of size at most $4|U| + k(4L+2)$ then by \cref{lm:SC_to_LCSR} there is a set cover of size at most $\lfloor \frac{4|U| + k(4L+2) - 4|U|}{4L+2} \rfloor = k$. This proves the reduction. From the \NP-hardness of set cover, we thus get the \NP-hardness of \LCSR.
    
We prove the logarithmic inapproximability as stated in \cref{thm:NP-hardness} by contraposition. Suppose there is a $c' \log M$ approximation algorithm $\mathcal A'$ for \LCSR for some constant $c' > 0$. We construct an algorithm $ \mathcal A$ that given an instance $I$ of set cover, uses the construction with $L = |U|$ to get an instance $(\mathcal G_1, \mathcal G_2)$ of \LCSR. Then the algorithm obtains a valid reconfiguration sequence $\sigma'$ on $I'$ with $\mathcal A'$, and computes a set cover $S'$ for $I$ as in \cref{lm:LCSR_to_SC}. We show that for some constant $c > 0$, $\mathcal A$ produces a $c \log(|U|+|\mathcal S'|)$ approximate solution to set cover.
By \cref{lm:SC_to_LCSR} there is a valid reconfiguration sequence between $\mathcal G_1$ and $\mathcal G_2$ of length at most $4|U| + OPT(I)(4L+2)$ where $OPT(I)$ is the minimum cardinality of a set cover of $U$. As our construction uses a polynomial number of vertices and edges, we have $M \le (|U| + |\mathcal S|)^d$ for some constant $d$. Therefore, $\mathcal A'$ finds a reconfiguration sequence $\sigma'$ of length at most $\ell := c' \log M \cdot (4|U| + OPT(I)(4L+2))$ in polynomial time. By \cref{lm:SC_to_LCSR}, we can compute a set cover $\mathcal S'$ of $U$ from $\sigma'$ with size
            \begin{align*}
            |\mathcal S'| \le \frac{\ell - 4 |U|}{4L+2} & =  \frac{c' \log M \cdot (4|U| + OPT(I)(4L+2)) - 4|U|}{4L+2 } \\
                & = c' \log M \cdot OPT(I) + \frac{4|U| \cdot ( c' \log M - 1)}{4L + 2}\\
                & \le c' \log M \cdot OPT(I) + c' \log M - 1\\
                & \le 2 c' \log M \cdot OPT(I)
            \end{align*}
Here we used $L = |U|$ in the first inequality to get $4|U|/(4L+2) \le 1$, and $OPT(I) \ge 1$ in the last inequality. Thus with $M \le (|U|+|\mathcal S|)^d$, we have $|\mathcal  S'| \le c \log (|U|+|\mathcal S|) OPT(I)$ for any constant $c > 2c'd$. Now we choose $c ' < c^*/2d$. This way we obtain a polynomial time $c^* \log (n+m)$ approximation algorithm for set cover, which completes the contraposition.
\end{proof}

\section{Conclusion}

In this paper, we started the research into temporal graph reconfiguration problems. As a first step, we introduced the $\LCRlong$ problem, which asks whether two temporal graphs can be transformed into each other through a sequence of edge relabeling operations while maintaining the always-connected property at every step. We showed how this problem is equivalent to the \STSRlong problem and answered the open question of \cite{DBLP:conf/isaac/HanakaIKOS24} by providing a simpler and graph-theoretic algorithm for it. We also showed a stronger hardness result than \cite{DBLP:conf/isaac/HanakaIKOS24}.

Beyond these results, several variants are worth considering. It would be interesting to consider properties of temporal graphs that have already been studied, such as source reachability and spanners, within the reconfiguration framework.

\bibliography{references.bib}

\newpage
\appendix

\section{Equivalence of LCR and STSR}

Hanaka et al. define the problem of \STSRlong \cite{DBLP:conf/isaac/HanakaIKOS24} as follows: Given a multigraph $G$, a sequence of $k$ spanning trees $(T_1, \dots T_k)$ is \emph{feasible} if they are edge-disjoint. Given two such feasible sequences, the goal is to reconfigure one sequence into another: In each step, a single spanning tree $T_i$ is replaced with $T_i'=T_i-e+f$ for $e\in E(T_i)$ and $f\in E(G)\setminus E(T_i)$.
This differs from \LCR in that each $T_i$ is a spanning tree as opposed to a connected graph, but there are also edges not contained in any spanning tree, and the reconfiguration step replaced an edge with such an unused edge.

We show that every instance of \LCR can be mapped to an equivalent instance of \STSR in linear time and vice versa. We denote the set of (static) edges of the $i$-th snapshot of a temporal graph as $E(\mathcal{G}(i)) := \{e \mid (e, i) \in \mathcal{E}(i)\}$. A spanning tree sequence $\mathbb{T} = (T_1, \cdots, T_k)$ over a multigraph $G$ is \textit{equivalent} to a temporal graph $\mathcal{G}$ of lifetime $k$ with $V(G) = V(\mathcal{G})$ and $E(G) = \biguplus_{i=1}^k E(\mathcal{G}(i))$ iff $E(T_i) \subseteq E(\mathcal{G}(i))$. Note that for every always-connected temporal graph we can construct a similar feasible spanning tree sequence as every snapshot contains a spanning tree. Similarly, given a feasible spanning tree sequence we can construct a similar always-connected temporal graph as the spanning trees are edge disjoint. We show the following lemmas for the equivalence of \LCR and \STSR.

\begin{lemma}
    Given a temporal graph $\mathcal{G}$ and two spanning tree sequences $\mathbb{T} = (T_1, \cdots, T_k)$ and $\mathbb{T}' = (T_1', \cdots, T_k')$ that are equivalent to $\mathcal{G}$, we can reconfigure $\mathbb{T}$ into $\mathbb{T}'$ in linear time.
\end{lemma}
\begin{proof}
    For every $1 \le i \le k$, $T_i$ and $T_i'$ are spanning trees of the graph induced by $E(\mathcal{G}(i))$. As every edge in $E(\mathcal{G}(i)) \setminus E(T_i)$ is not part of any spanning tree $T_j$ for $1 \le j \le k$, the reconfiguration of $\mathbb{T}$ into $\mathbb{T}'$ reduces to the spanning tree reconfiguration problem for every spanning tree in the sequence. This can be solved in linear time as observed in \cite{DBLP:journals/tcs/ItoDHPSUU11}.
\end{proof}
\begin{lemma}
    Given a spanning tree sequence $\mathbb{T}$ and two temporal graphs $\mathcal{G}$, $\mathcal{G}'$ that are equivalent to $\mathbb{T}$, we can reconfigure $\mathcal{G}$ into $\mathcal{G}'$ in linear time.
\end{lemma}
\begin{proof}
    For every $1 \le i \le k$ the snapshots $\mathcal{G}(i)$ and $\mathcal{G}'(i)$ both contain $E(T_i)$ as they are equivalent. Thus any differing edges can be changed in a single operation as connectivity is ensured by the spanning trees.
\end{proof}

\begin{lemma}
    Given a temporal graph $\mathcal{G}$ and an equivalent spanning tree sequence $\mathbb{T}$ and an operation on one of them, we can do one operation on the other object to maintain equivalence.
\end{lemma}
\begin{proof}
    Let $T_i' = T_i - e + f$ be some operation on $\mathbb{T}$ for some $1 \le i \le k$ with $e \in E(T_i)$ and $f \in E(G) \setminus E(T_i)$. Since $f$ is not part of $E(T_j)$ for $1 \le j \le k$, we can relabel the corresponding temporal edge in $\mathcal{G}$ to $(f, i)$ while ensuring connectivity in $\mathcal{G}(j)$. This maintains equivalence because $T_i' \subseteq E(\mathcal{G}(i))$ holds after relabeling.

    Let $(e,i) \to (e, j)$ be a valid relabeling operation in $\mathcal{G}$. If $e \not \in E(T_i)$ holds they are still equivalent. Thus we assume that $e \in E(T_i)$ is true. As $\mathcal{G}(i) - e$ is connected, there exists an edge $f \in E(\mathcal{G}(i))$ such that $T_i - e + f$ is a spanning tree of $G$. Since $f$ is in $E(\mathcal{G}(i))$, but not in $E(T_i)$ we know that $f$ is not part of any spanning tree of $\mathbb{T}$. Therefore, the operation $T_i' = T_i - e + f$ is valid and maintains equivalence.
\end{proof}

\end{document}